\documentclass[twocolumn,pra,showpacs,10pt]{revtex4-1}
\usepackage{graphicx,amsmath,amssymb,xspace,ntheorem,hyperref}
\usepackage{color}
\usepackage[T1]{fontenc}

\newtheorem{definition}{Definition}
\newtheorem{theorem}{Theorem}
\newtheorem{lemma}{Lemma}

\newtheorem*{proof}{Proof}

\begin{document}

\newcommand{\ketbra}[2]{| #1\rangle \langle #2|}
\newcommand{\ket}[1]{| #1 \rangle}
\newcommand{\bra}[1]{\langle #1 |}
\newcommand{\Tr}{\mathrm{Tr}}
\newcommand\F{\mbox{\bf F}}
\newcommand{\h}{\mathcal{H}}

\newcommand{\PSD}{\textup{PSD}}

\newcommand{\C}{\mathbb{C}}
\newcommand{\X}{\mathcal{X}}
\newcommand{\Y}{\mathcal{Y}}
\newcommand{\Z}{\mathcal{Z}}
\newcommand{\sspan}{\mathrm{span}}
\newcommand{\kb}[1]{\ket{#1} \bra{#1}}
\newcommand{\pos}{D}

\newcommand{\thmref}[1]{\hyperref[#1]{{Theorem~\ref*{#1}}}}
\newcommand{\lemref}[1]{\hyperref[#1]{{Lemma~\ref*{#1}}}}
\newcommand{\corref}[1]{\hyperref[#1]{{Corollary~\ref*{#1}}}}
\newcommand{\eqnref}[1]{\hyperref[#1]{{Equation~(\ref*{#1})}}}
\newcommand{\claimref}[1]{\hyperref[#1]{{Claim~\ref*{#1}}}}
\newcommand{\remarkref}[1]{\hyperref[#1]{{Remark~\ref*{#1}}}}
\newcommand{\propref}[1]{\hyperref[#1]{{Proposition~\ref*{#1}}}}
\newcommand{\factref}[1]{\hyperref[#1]{{Fact~\ref*{#1}}}}
\newcommand{\defref}[1]{\hyperref[#1]{{Definition~\ref*{#1}}}}
\newcommand{\exampleref}[1]{\hyperref[#1]{{Example~\ref*{#1}}}}
\newcommand{\hypref}[1]{\hyperref[#1]{{Hypothesis~\ref*{#1}}}}
\newcommand{\secref}[1]{\hyperref[#1]{{Section~\ref*{#1}}}}
\newcommand{\chapref}[1]{\hyperref[#1]{{Chapter~\ref*{#1}}}}
\newcommand{\apref}[1]{\hyperref[#1]{{Appendix~\ref*{#1}}}}
\newcommand\rank{\mbox{\tt {rank}}\xspace}
\newcommand\prank{\mbox{\tt {rank}$_{\tt psd}$}\xspace}
\newcommand\alice{\mbox{\sf Alice}\xspace}
\newcommand\bob{\mbox{\sf Bob}\xspace}
\newcommand\pr{\mbox{\bf Pr}}
\newcommand\av{\mbox{\bf{\bf E}}}
\newcommand{\pabxy}{(p(ab|xy))}
\newcommand{\calQ}{\mathcal{Q}}
\def\be{\begin{equation}}
\def\ee{\end{equation}}

\newcommand{\comment}[1]{{}}
\newcommand{\blue}[1]{\textcolor{blue}{#1}}
\newcommand{\red}[1]{\textcolor{red}{#1}}

\title{\vspace{-1cm} Analytic Semi-device-independent Entanglement Quantification for Bipartite Quantum States}

\author{Zhaohui Wei$^{}$}\email{Email: weizhaohui@gmail.com}
\author{Lijinzhi Lin$^{}$}
\affiliation{$^{}$Center for Quantum
Information, Institute for Interdisciplinary
Information Sciences, Tsinghua University, Beijing 100084, China}

\begin{abstract}
We define a property called nondegeneracy for Bell inequalities, which describes the situation that in a Bell setting, if a Bell inequality and involved local measurements are chosen and fixed, any quantum state with a given dimension and
its orthogonal quantum state cannot violate the inequality remarkably at the same time. By choosing a proper nondegenerate Bell inequality, we prove that for a unknown bipartite quantum state of a given dimension, based on the measurement statistics only, we can provide an analytic lower bound for the entanglement of formation or even for the distillable entanglement, making the whole process semi-device-independent. We characterize the mathematical structure of nondegenerate Bell inequalities, and prove that quite a lot of well-known Bell inequalities are nondegenerate. We demonstrate our approach by quantifying entanglement for qutrit-qutrit states based on their violation to the CGLMP inequality.
\end{abstract}

\maketitle

\section{Introduction}

It has been well-known that entanglement is a major computational resource in quantum information processing and quantum communication tasks, thus certifying entanglement for a unknown quantum system reliably in quantum labs is a fundamental and important problem. For small quantum systems tomography is a possible solution \cite{CN97,PCZ97}, but as the problem size grows, the cost of tomography goes up exponentially, making this approach infeasible. In this case, one can instead use the idea of entanglement witness to detect entanglement \cite{GT09}, but one drawback of this approach is that the knowledge on quantum dimension and the accurate measurement implementations must be given, which are often unpractical, otherwise the results may not be reliable \cite{RFB+12}.

To overcome this problem, it turns out that the approach of device-independence, a method that was first introduced in the area of quantum key distribution \cite{Ekert91,BHK05,AGM06} and self-testing \cite{PR92,MY98}, is very helpful. In this approach, all involved quantum devices are regarded as black boxes and quantum tasks like entanglement certification are usually accomplished by checking the existence of Bell nonlocality, i.e., a violation to some Bell inequality that any classical systems cannot make \cite{Bell64}. Particularly, this approach has been utilized extensively to certify the existence of genuine multipartite entanglement \cite{CGP+02,BGLP11,PV11,MBL+13,MRMT16,BCWA17,TARGB18,ZDBS19}. Since nontrivial and reliable conclusions can be drawn from limited measurement data only, device-independence is highly valuable experimentally. Moreover, for the situations that partial reliable information on the target quantum systems is known, people add some modest assumptions to fully device-independent quantum models, resulting in measurement-device-independent \cite{LCQ12,BP12} and semi-device-independent scenarios \cite{LVB11,MG12}.

A further step from entanglement certification is the quantification of entanglement in quantum labs \cite{IMP+15,BEJ+19,LJF+18}. In order to provide reliable results, device-independent schemes for quantifying entanglement have also been proposed. Inspired by the NPA method \cite{NPA07}, a device-independent method to lower bound the negativity was provided in \cite{MBL+13}. Using the concept of semiquantum nonlocal games introduced in \cite{Buscemi12}, 
a measurement-device-independent approach to quantify negative-partial-transposition entanglement has also been reported \cite{SHR17}. Usually, this kind of works face two inevitable difficulties. First, nonlocality and entanglement are known as two different resources for quantum information processing \cite{MS07}, profoundly making quantifying entanglement in a device-independent way challenging. Second, the mathematical structures of sets of quantum correlations is very complicated \cite{BCP+14,GKW+18,DW15,Slof17}, for example accurate Tsirelson bounds are often notoriously hard to find out, which makes it quite hard to study most device-independent quantum tasks in an analytical way, especially when the dimension is high. As a consequence, in most cases of device-independent quantum tasks one needs to perform costly numerical calculations \cite{NPA07}. Therefore, despite these encouraging progresses, in order to gain deeper understanding for the fundamental relations between nonlocality and entanglement measures, especially those standard entanglement measures with clear operational meanings, direct \emph{analytical} results for general cases of Bell experiments are highly demanded.

In this paper, for a general unknown bipartite quantum state, we provide an analytic method to quantify the entanglement of formation or the distillable entanglement, two of the most well-known standard entanglement measures, in a semi-device-independent manner, where besides the measurement statistics data, the only assumption we make is quantum dimension. The main idea behind our approach is a new property called nondegeneracy we define for Bell inequalities. Basically, in a Bell setting, if any quantum state $\ket{\psi}$ of a given dimension and any quantum state orthogonal to $\ket{\psi}$ cannot violate the inequality remarkably at the same time by using the same set of local measurements, we say the Bell inequality is nondegenerate. By looking into the mathematical structure of nondegenerate Bell inequalities, we prove that a lot of well-known Bell inequalities are nondegenerate, including the CHSH inequality, the $I_{3322}$ inequalty, and the CGLMP inequalities. Actually we conjecture that most of Bell inequalities satisfy this property. By choosing nondegenerate Bell inequalities, we prove that a fundamental relation between Bell inequality violations and the entanglement measures can be built, eventually giving the analytic result we want. We demonstrate the applications of our approach by applying the CGLMP inequalities on qutrit-qutrit quantum states, and specific examples show that nontrivial lower bound for the entanglement measures can be obtained as long as the violation is sufficient.

\section{Nondegenerate Bell Inequalities}

In a two-party Bell experiment, Alice and Bob located at different places share a physical system, and perform local measurements on their own subsystems without communications.
Specifically, Alice (Bob) has a set of measurement apparatus labelled by a finite set $X$ ($Y$), and the set of possible measurement outcomes are labelled by a finite set $A$ ($B$). When the experiment begins, they choose random apparatuses
to measure the system and repeat the whole process many times. By recording the frequency of every choices and corresponding outcomes, they calculate the joint probability distribution $p(ab|xy)$, indicating the probability of obtaining outcomes
$a\in A$ and $b\in B$ when choosing measurement apparatuses $x\in X$ and $y\in Y$. The collection of all $|A\times B\times X\times Y|$ joint probability distributions can be written as a vector $p:=\{p(ab|xy)\}$, called a correlation.

The set of correlations depends heavily on the physical laws that the system that Alice and Bob share obeys. If the experiment is purely classical, all the correlations they are able to produce are local correlations, which can be explained by sharing a public randomness before the experiment begins, and then generating local distributions with respect to the distribution of the public randomness, which is called a local hidden variable (LHV) model. On the other hand, if what they share beforehand is a quantum state $\rho$, the correlation is called quantum and can be written as,
\begin{equation}
p(ab|xy)=\Tr(\rho(M_x^a\otimes M_y^b)),
\end{equation}
where $M_x^a$ and $M_y^b$ are the measurement operators of the measurement apparatuses $x$ and $y$.

A major discovery of quantum mechanics is that there exist quantum correlations that cannot be produced with LHV models, which can be explained by the concept of Bell inequalities \cite{Bell64}. A typical Bell inequality can be expressed as
\begin{equation}
I:=\sum_{abxy}s_{abxy}p(ab|xy)\leq C_l,
\end{equation}
where $s_{abxy}$ are normally real coefficients, and $C_l$ is the maximal value of the Bell expression $I$ that local correlations achieve. It turns out that in some cases the maximal value of $I$ that quantum correlations achieve, called Tsirelson bound and denoted $C_q$, can be strictly larger than $C_l$, revealing the profound discovery we just mentioned.

In this paper, we define and focus on a special case of Bell inequalities called {\em nondegenerate}. We will show that for a unknown bipartite quantum state $\rho$, this property makes it possible to obtain analytic results on the entanglement of formation, denoted $E_f(\rho)$, by utilizing the measurement statistics data only, assuming the quantum dimension is known.

For convenience, we denote the Bell expression of the correlation generated by measuring a quantum state $\rho$ acting on Hilbert space $\h^{d} \otimes \h^{d}$ with measurements $\{M_x^a\}$ and $\{M_y^b\}$ as $I(\rho,M_x^a,M_y^b)$. Then we have the following definition.
\begin{definition}
A Bell inequality $I\leq C_l$ is {\em nondegenerate}, if there exists two real number $0\leq\epsilon_1<\epsilon_2\leq C_q$, such that for any pure state $\ket{\psi}$ acting on $\h^{d} \otimes \h^{d}$ and any measurements $\{M_x^a\}$ and $\{M_y^b\}$,
\begin{equation*}
I(\ket{\psi}\bra{\psi},M_x^a,M_y^b)\geq C_q-\epsilon_1
\end{equation*}
always implies that
\begin{equation*}
I(\ket{\psi^\bot}\bra{\psi^\bot},M_x^a,M_y^b)\leq C_q-\epsilon_2,
\end{equation*}
where $\ket{\psi^\bot}$ is any pure state orthogonal to $\ket{\psi}$.
\end{definition}

Intuitively, the nondegeneracy of a Bell inequalities means that if a quantum state makes a large violation to the Bell inequality, any orthogonal quantum state cannot with the involved measurements {\em unchanged}.

A few remarks on this definition are in order. First, nondegeneracy is meaningful only when the dimension is given, as any Bell inequality cannot satisfy the definition if extra dimensions can be introduced freely in the form of ancillary subsystems. Second, note that in some device-independent quantum tasks like self-testing \cite{PR92,BMR92,MY98,CGS17}, a crucial issue is whether the maximal violation to a Bell inequality is achieved by multiple pure quantum states, where the involved measurement sets can be essentially different. For convenience in this case we say this Bell inequality enjoys the uniqueness property. We stress that the nondegeneracy property we define is much weaker than the uniqueness property. After all, in principle it is possible that two close but essentially different quantum pure states achieve the maximal violation at the same time, but they are using different measurements, thus still satisfy the definition of nondegeneracy. Usually it is notoriously hard to determine whether or not a given Bell inequality has the uniqueness property. Therefore, a looser requirement in the definition may make it much easier to certify the nondegeneracy property, which potentially results in wider applications of this new definition. Actually, later we will see that quite a lot of well-known Bell inequalities are nondegenerate with dimension restricted. Third, another issue worth pointing out is that, though for simplicity we mainly focus on linear forms of Bell inequalities in this paper, nondegeneracy can be defined on general Bell inequalities of any forms.

\section{Principal component analysis}

Before proving that nondegenerate Bell inequalities do exist, let us see that in Bell experiments, the property of nondegeneracy can provide useful information on the purity of the underlying shared quantum states.

Suppose in a Bell experiment, a quantum correlation $p(ab|xy)$ is produced by measuring a bipartite quantum state $\rho$ of dimension $d\times d$, where the involved measurements are $\{M_x^a\}$ and $\{M_y^b\}$. 
We suppose there exists a nondegenerate Bell inequality $I\leq C_l$ with parameters $\epsilon_1$ and $\epsilon_2$ such that the Bell expression given by $p(ab|xy)$ is larger than $C_q-\epsilon_1$, that is,
\begin{equation}\label{eq:assumption}
I(\rho,M_x^a,M_y^b)\geq C_q-\epsilon_1.
\end{equation}

Intuitively, if $\epsilon_1$ is very small, then usually the quantum state $\rho$ that produces $p(ab|xy)$ is very close to a pure state that maximizes $I(\rho,M_x^a,M_y^b)$. Actually, this intuition can be captured well using the concept of nondegeneracy.
Let an orthogonal decomposition of $\rho$ be $\rho=\sum_{i=1}^{d^2}a_i\ket{\psi_i}\bra{\psi_i}$. Since for fixed local measurements the Bell expression is linear in the shared quantum state, there must be a $\ket{\psi_i}$ such that $I(\ket{\psi_i}\bra{\psi_i},M_x^a,M_y^b)\geq C_q-\epsilon_1$. Without loss of generality, we suppose $i=1$. Then it holds that
\begin{eqnarray*}
I(\rho,M_x^a,M_y^b)&=&\sum_{i=1}^{d^2}a_i\cdot I(\ket{\psi_i}\bra{\psi_i},M_x^a,M_y^b)\\
&\leq & a_1\cdot I(\ket{\psi_1}\bra{\psi_1},M_x^a,M_y^b)\\
&&+(1-a_1)(C_q-\epsilon_2)\\
&\leq & a_1\cdot C_q+(1-a_1)(C_q-\epsilon_2),
\end{eqnarray*}
where we have used the definition of nondegenerate Bell inequality and the fact that the maximal Bell expression for quantum correlations is $C_q$.

Combining the above inequality with Eq.\eqref{eq:assumption}, we immediately have that
\[
a_1\geq1-\epsilon_1/\epsilon_2.
\]
Therefore, if $\epsilon_1/\epsilon_2\ll1$, the nondegeneracy guarantees that violating the Bell inequality almost maximally means that the involved quantum state $\rho$ must be close to pure, as the purity of $\rho$ can be lower bounded by
\[
\Tr(\rho^2)=\sum_{i=1}^{d^2}a_i^2\geq \left(1-\frac{\epsilon_1}{\epsilon_2}\right)^2+\frac{(\epsilon_1/\epsilon_2)^2}{d^2-1}.
\]
In this case, the principal component of $\rho$ we define is also the component of $\rho$ with the largest weight in the orthogonal decomposition.

\section{The certification of nondegeneracy}

We now show that the concept of nondegeneracy Bell inequalities is well-defined, and a lot of well-known Bell inequalities are indeed nondegenerate. In fact, we conjecture that most Bell inequalities satisfy this definition.

			Consider a Bell scenario over finite setting sets $\mathcal{X}$, $\mathcal{Y}$ and finite outcome sets $\mathcal{A}$, $\mathcal{B}$.
			Then a {Bell expression} $I$ is a function from the set of bipartite states to $\mathbb{R}$
			admitting the form
			\begin{align*}
				I(\rho_{AB},M_x^a,M_y^b)= & \sum_{abxy} s_{abxy}\mathrm{Tr}((M_x^a\otimes M_y^b)\rho_{AB}),
			\end{align*}
			where $s_{abxy}\in\mathbb{R}$, $a\in \mathcal{A}$, $b\in \mathcal{B}$, $x\in \mathcal{X}$, $y\in\mathcal{Y}$, and $\{M_x^a\}$ and $\{M_y^b\}$ are POVMs on $d$-dimensional quantum subsystems $A$ and $B$, respectively.
			In particular, for a pure state $\ket{\psi}_{AB}$, if we let
\[
H(M_x^a,M_y^b)=\sum_{abxy} s_{abxy}M_x^a\otimes M_y^b,
\]
then we have
			\begin{align*}
				I(\ket{\psi}_{AB}\bra{\psi}_{AB},M_x^a,M_y^b)= \bra{\psi}_{AB}H(M_x^a,M_y^b)\ket{\psi}_{AB}.
			\end{align*}
		Since $H(M_x^a,M_y^b)$ is Hermitian, it has $d^2$ real eigenvalues, and we now denote them by $\lambda_1(H(M_x^a,M_y^b))\geq \cdots\geq \lambda_{d^2}(H(M_x^a,M_y^b))$. Furthermore, we define
			\begin{align*}
				C(I,d,t)= & \max_{\{M_x^a,M_y^b\}}\sum\limits_{i=1}^{t}\lambda_i(H(M_x^a,M_y^b)).
			\end{align*}
Then it is not hard to see that $C_q=C(I,d,1)$.


		We now show that there is a simple relation between $C(I,d,k)$ and nondegeneracy of $I$, as shown in the following lemma.

		\begin{lemma}\label{lemma:parameter}
			For any bipartite quantum system of dimension $d\times d$, a Bell expression $I$ is nondegenerate with $0\leq\epsilon_1<\epsilon_2$ if and only if $C(I,d,2)<2C(I,d,1)$.
		\end{lemma}
		\begin{proof}
			Suppose $I$ is nondegenerate with $0\leq\epsilon_1<\epsilon_2$.
			Suppose POVMs $\{M_x^a\}$ and $\{M_y^b\}$ maximize $C(I,d,2)$.
			And let $\ket{\psi_1},\ket{\psi_2}$ be the eigenstates corresponding to $\lambda_1(H(M_x^a,M_y^b))$ and $\lambda_2(H(M_x^a,M_y^b))$ respectively.
			If $C(I,d,2)=2C(I,d,1)$, then, by
			\begin{align*}
				C(I,d,2)= & \lambda_1(H(M_x^a,M_y^b))+\lambda_2(H(M_x^a,M_y^b)) \\
				= & I(\ket{\psi_1}\bra{\psi_1},M_x^a,M_y^b)+I(\ket{\psi_2}\bra{\psi_2},M_x^a,M_y^b),
			\end{align*}
			we have
			\begin{align*}
				I(\ket{\psi_1}\bra{\psi_1},M_x^a,M_y^b)=I(\ket{\psi_2}\bra{\psi_2},M_x^a,M_y^b)=C_q,
			\end{align*}
			which contradicts with nondegeneracy parameters $\epsilon_1<\epsilon_2$.

			Conversely, suppose $C(I,d,2)<2C(I,d,1)$.
			For any pair of orthogonal pure states $\ket{\psi}$, $\ket{\phi}$ and any POVMs $\{M_x^a\}$ and $\{M_y^b\}$, we have
			\begin{align*}
				I(\ket{\psi}\bra{\psi},M_x^a,M_y^b)+I(\ket{\phi}\bra{\phi},M_x^a,M_y^b)\leq C(I,d,2).
			\end{align*}
			Then, if $I(\ket{\psi}\bra{\psi},M_x^a,M_y^b)\geq C_q-\epsilon_1$ for some small positive number $\epsilon_1$, it can be verified that
			\begin{align*}
				I(\ket{\phi}\bra{\phi},M_x^a,M_y^b)\leq & C(I,d,2)-I(\ket{\psi}\bra{\psi},M,N) \\
				\leq & C(I,d,2)-C_q+\epsilon_1\\
=&C_q-[2C_q-C(I,d,2)-\epsilon_1].
			\end{align*}
Therefore, if we let $\epsilon_1<C_q-\frac{1}{2}C(I,d,2)$ and $\epsilon_2=2C_q-C(I,d,2)-\epsilon_1$, then we have that $\epsilon_1<\epsilon_2$ and $I(\ket{\phi}\bra{\phi},M_x^a,M_y^b)\leq C_q-\epsilon_2$,
			which implies that $I$ is nondegenerate with parameters $\epsilon_1$ and $\epsilon_2$.
		\end{proof}

		We have two ways to compute $C(I,d,2)$.
		One way is to directly maximize $C(I,d,2)$ as a convex function over the convex set of POVMs.
The maximum value is achieved at the extreme points, that is, the extremal POVM, as pointed out in \cite{DPP05}.
		Additionally, according to \cite{DPP05}, for the case of qubits ($d=2$), extremal POVM have outcome numbers at most 4 and consist of rank 1 projectors when the number of outcomes is at least 2, further simplifying the optimization.

		We next introduce another way to estimate $C(I,d,2)$, which resorts to $C(I,d-1,1)$, i.e., the computation of another Tsirelson bound with a smaller dimension. For this, we first show that, in a $d\times d$ bipartite system, we can ``extract'' a state with Schmidt number at most $d-1$ by linearly combining two states with Schmidt number exactly $d$.
		\begin{lemma}
			\label{lemma:comb}
			Let $\ket{\psi}$, $\ket{\phi}$ be two bipartite states in $d\times d$ dimensional system.
			If both of $\ket{\psi}$ and $\ket{\phi}$ have Schmidt number $d$, then there is $\alpha,\beta\in\mathbb{C}$ with $\alpha\beta\neq 0$
			such that $\alpha\ket{\psi}+\beta\ket{\phi}$ have Schmidt number at most $d-1$.
		\end{lemma}
		\begin{proof}
			For any state $\ket{\varphi}=\sum_{ij}a_{ij}\ket{i}\otimes\ket{j}$, we transform it into a $d\times d$ matrix with $(i,j)$-th entry equal to $a_{ij}$.
			We transform $\ket{\psi}$ and $\ket{\phi}$ into $A$ and $B$ in this fashion, respectively.
			Then both $A$ and $B$ have full rank, that is, rank $d$.

			The linear combination of $A$ and $B$ reads
			\begin{align*}
				\alpha A+\beta B= & A(\alpha I+\beta A^{-1}B).
			\end{align*}
			Let $C=A^{-1}B$; then $C$ has full rank as well.
			By assumingt that $\beta\neq 0$, we can write
			\begin{align*}
				\alpha A+\beta B= & \beta A(\gamma I+C),
			\end{align*}
			where $\alpha/\beta=\gamma\in\mathbb{C}$ is arbitary.
			Since $C$ is a complex matrix, it has a nonzero eigenvalue $\lambda$; that is, $C-\lambda I$ is of rank at most $d-1$.
			By picking $\gamma=-\lambda$, the resulting linear combination $\alpha A+\beta B$ has rank at most $d-1$,
			so the Schmidt number of $\alpha \ket{\psi}+\beta \ket{\phi}$ is at most $d-1$ as well.
		\end{proof}

Then we have the following characterization of nondegeneracy for Bell inequalities.
		\begin{theorem}
			\label{theo:dimwit}
			Let $I$ be a Bell expression and $d> 1$.
			If $C(I,d,1)>C(I,d-1,1)$, then $I$ is nondegenerate.
		\end{theorem}
		\begin{proof}
We now prove that if $C(I,d,1)>C(I,d-1,1)$, then $C(I,d,2)\leq C(I,d,1)+C(I,d-1,1)$. According to Lemma \ref{lemma:parameter}, this implies that $I$ is nondegenerate.
			Suppose $C(I,d,2)>C(I,d,1)+C(I,d-1,1)$.
			Let $\{M_x^a\}$ and $\{M_y^b\}$ be the POVMs that achieve $C(I,d,2)$.
			Then there exist two corresponding eigenstates $\ket{\psi}$, $\ket{\phi}$ satisfy
			\begin{align*}
				&I(\ket{\psi}\bra{\psi},M_x^a,M_y^b)+I(\ket{\phi}\bra{\phi},M_x^a,M_y^b)\\
=&C(I,d,2)\\
> & C(I,d,1)+C(I,d-1,1).
			\end{align*}
			By definition of $C(I,d,1)$, this means that
			\begin{align*}
				I(\ket{\psi}\bra{\psi},M_x^a,M_y^b)> & C(I,d-1,1), \\
				I(\ket{\phi}\bra{\phi},M_x^a,M_y^b)> & C(I,d-1,1),
			\end{align*}
			thus both $\ket{\psi}$ and $\ket{\phi}$ have Schmidt number at least $d$.
			For $\alpha,\beta\in\mathbb{C}$ with $|\alpha|^2+|\beta|^2=1$, we have
			\begin{align*}
				& I((\alpha\ket{\psi}+\beta\ket{\phi})(\bar{\alpha}\bra{\psi}+\bar{\beta}\bra{\phi}),M_x^a,M_y^b) \\
				= & |\alpha|^2 I(\ket{\psi}\bra{\psi},M_x^a,M_y^b)+|\beta|^2 I(\ket{\phi}\bra{\phi},M_x^a,M_y^b)\\
>& C(I,d-1,1).
			\end{align*}
			However, by Lemma \ref{lemma:comb}, there is $\alpha,\beta\in\mathbb{C}$ with $\alpha\beta\neq 0$
			such that $\ket{\varphi}=\alpha\ket{\psi}+\beta\ket{\phi}$ has Schmidt number at most $d-1$.
			By normalizing the linear combination, we can fit $\ket{\varphi}$ into a $(d-1)\times (d-1)$ dimensional system.
			Let ${M'}_x^a$, ${M'}_y^b$ be the compression of $M_x^a$ and $M_y^b$ into the reduced system.
			Then we have
			\begin{align*}
				I(\ket{\varphi}\bra{\varphi},{M'}_x^a,{M'}_y^b)>C(I,d-1,1),
			\end{align*}
			which contradicts with the definition of $C(I,d-1,1)$.
			Therefore, $C(I,d,2)\leq C(I,d,1)+C(I,d-1,1)$.
		\end{proof}

This theorem implies the following two interesting consequences. First, any Bell inequality that can be violated by a pair of qubits is nondegenerate with $0\leq\epsilon_1<\epsilon_2$.
				Indeed, when $d=1$, the system is entirely classical, hence $C(I,1,1)<C(I,2,1)$.
				In particular, the CHSH inequality is nondegenerate, and actually quite a lot of device-independent characterization of qubit-qubit states based on the CHSH inequalities have been reported \cite{SG01,BLM+09,Jed16}.

Second, any Bell expression with its Tsirelson bound strictly monotonic with respect to $d$ is nondegenerate at any dimension.
				Two well-known Bell inequalities with this property are the $I_{3322}$ inequality and the CGLMP inequality \cite{ZG08,PV10}. Therefore, both of them are nondegenerate (at least for certain dimensions).

\section{An example: quantifying entanglement with the CGLMP inequality}

We now show that the concept of nondegenerate Bell inequalities allows us to analytically quantify the entanglement of a unknown bipartite quantum state in a semi-device-independent manner. For simplicity, we will focus on the CGLMP inequality for a qutrit-qutrit quantum state $\rho$ acting on Hilbert space $\h^{3} \otimes \h^{3}$. We stress that our approach can be applied generally on quantum states of any dimension and any nondegenerate Bell inequalities.

The form of the CGLMP inequality we choose is essentially from \cite{ZG08}, which can be expressed as
		\begin{align*}
			\mathrm{P}(A_2\geq B_2)+&\mathrm{P}(B_2\geq A_1)+\\
&\mathrm{P}(A_1\geq B_1)+\mathrm{P}(B_1>A_2)\leq 3.
		\end{align*}
In \cite{ZG08}, it has been found out that when $d=3$, $C_q=C(I,3,1)=3.3050$. Through numerical simulations, we find that for qutrit-qutrit quantum states, $C(I,3,2)=6.2071$. Note that $C(I,3,2)<2\cdot C(I,3,1)$, then Lemma \ref{lemma:parameter} indicates that the CGLMP inequality is nondegenerate for $d=3$. Furthermore, the proof to Lemma \ref{lemma:parameter} provides a systematic way to choose the corresponding parameters $\epsilon_1$ and $\epsilon_2$. Therefore, if a target quantum state $\rho$ satisfies that $I(\rho,M_x^a,M_y^b)\geq C_q-\epsilon_1$, we can use the principal component analysis introduced before to obtain a lower bound for the purity of $\rho$, that is,
\[
\Tr(\rho^2)\geq \left(1-\frac{\epsilon_1}{\epsilon_2}\right)^2+\frac{(\epsilon_1/\epsilon_2)^2}{8}\equiv\gamma_{\rho}.
\]
Then according to \cite{SSY+17}, the Von Neumann entropy of $\rho$, denoted by $S(\rho)$, can be upper bounded as
\[
S(\rho)\leq -c_i\sum_{i=1}^9\log(c_i),
\]
where $c_1=\frac{1}{9}-\frac{2}{3}\sqrt{2(\gamma_{\rho}-\frac{1}{9})}$, and $c_2=\cdots=c_9=(1-c_1)/8$.


On the other hand, according to \cite{SVV15,WS17}, the purity of $\rho_A=\Tr_B(\rho)$ (or $\rho_B=\Tr_A(\rho)$) can also be upper bounded. Indeed, define
\begin{small}
\begin{equation}\label{eq:fun1}
f_1(p)=\min_{y_1,y_2} \sum_{b_1,b_2} \min_x \bigg(\sum_a \sqrt{p(ab_1|xy_1) p(ab_2|xy_2)} \bigg)^2,
\end{equation}
and
\begin{equation}\label{eq:fun2}
f_2(p)=\min_{x_1,x_2} \sum_{a_1,a_2} \min_y \bigg(\sum_b \sqrt{p(a_1b|x_1y) p(a_2b|x_2y)} \bigg)^2,
\end{equation}
\end{small}
then it holds that \cite{SVV15,WS17}
\begin{equation}
\Tr(\rho_{A}^2)\leq \min\{f_1(p),f_2(p)\}\equiv \gamma_A.
\end{equation}
Again, when $\gamma_A<1/2$, according to \cite{SSY+17} the Von Neumann entropy of $\rho_A$ can be lower bounded as
\[
S(\rho_A)\geq -f_i\sum_{i=1}^3\log(f_i),
\]
where $f_1=f_2=\frac{1-\alpha}{2}$, $f_3=\alpha$, and $\alpha=\frac{1}{3}-\sqrt{\frac{2}{3}(\gamma_{A}-\frac{1}{3})}$.

We next consider the coherent information of $\rho$ defined as \cite{SN96,Lloyd97}
\[
I_C(\rho)=S(\rho_A)-S(\rho).
\]
Clearly, our discussions above provide an analytical approach to lower bound the coherent information of $\rho$.

Importantly, it turns out that, for any bipartite quantum state $\rho$, we have that \cite{COF11}
\begin{equation}
E_f(\rho)\geq E_D(\rho)\geq I_C(\rho),
\end{equation}
where $E_D(\rho)$ is the distillable entanglement of $\rho$. Therefore, our approach actually lower bounds the entanglement of formation or even the distillable entanglement for $\rho$. Note that in addition to the measurement statistics $p(ab|xy)$, we do not need any assumption on the internal working of the quantum system or the precision of quantum operations except the system dimension $d$, which means that our quantification for $E_f(\rho)$ or $E_D(\rho)$ is of a semi-device-independent nature.

We test our approach on numerically generated qutrit-qutrit correlations, and the results are illustrated in the figure below. It can be seen that when the gap between the violation and $C_q$ is smaller than $0.07$, our method gives positive lower bound for the distillable entanglement.
\begin{figure}[htbp]
   \label{fig:example}
   \centering
   \includegraphics[width=3.25in]{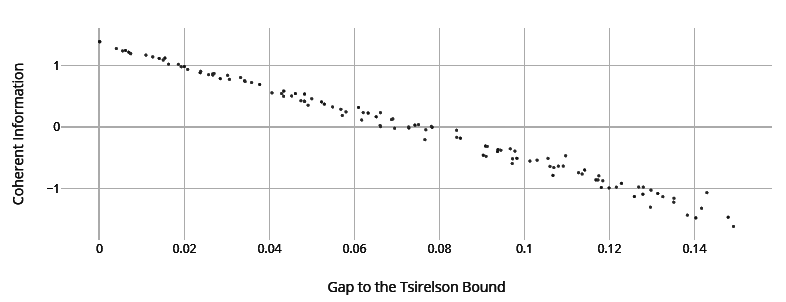}
   \caption{{Our lower bounds for the coherent information (or the distillable entanglement) based on the violations to the CGLMP inequality, where quantum correlations are generated by measuring random qutrit-qutrit states. Note that the gap between the classical bound and the Tsirelson bound is 0.3050.}
   }
\end{figure}

Lastly, we would like to point out that $E_f(\rho)$ can also be lower bounded by the following alternative way. According to \cite{WS17}, $E_f(\ket{\psi_1}\bra{\psi_1})$ can be lower bounded as the purity of $\Tr_B(\ket{\psi_1}\bra{\psi_1})$ can be upper bounded, where $\ket{\psi_1}$ is the principal component of $\rho$ we have discussed above. Then by the continuous property of the entanglement of formation proved by \cite{Nielsen00,Winter16}, we can bound the gap between $E_f(\rho)$ and $E_f(\ket{\psi_1}\bra{\psi_1})$. Combining these two results together, we can obtain a lower bound for $E_f(\rho)$. However, specific examples of quantum correlations show that our first approach is much better than the second one.

\section{Multipartite case}

In principle the approach above can be generalized to multipartite case \cite{WS19}, as the concept of nondegeneracy can also be defined naturally on multipartite Bell inequalities. But a major issue are raised in multipartite case and has to be addressed, which is the structure of multipartite entanglement is much more complicated. For example, because of the existence of Schmidt decompositions, entanglement quantification for bipartite pure states based on measurement statistics data can be achieved as addressed in \cite{WS17}, but Schmidt decompositions do not always exist for multipartite pure quantum states, thus this part has to be redeveloped carefully. Similarly, bounding coherent information or applying the continuous property of the entanglement of formation will be much more challenging in multipartite case. Anyway, we hope these problems will be discussed in future work.

\section{Conclusions}

In order to develop method that is able to quantify entanglement reliably in quantum labs, in this paper we define a property called nondegeneracy for Bell inequalities. We believe that this property is of independent interest, and provides us a new insight to study Bell inequalities, a fundamental and important tool in quantum physics and quantum information. Based on the concept of nondegenerate Bell inequalities, we propose an approach to quantify the entanglement of formation or the distillable entanglement for the shared quantum state underlying a Bell experiment in a semi-device-independent manner, which is analytic and does not rely on complicated numerical optimizations, unlike most results on device-independent quantum tasks. We also provide a mathematical characterization for nondegenerate Bell inequalities, and prove that quite a lot of well-known Bell inequalities are nondegenerate. We apply our approach on qutrit-qutrit quantum states by choosing the CGLMP inequality, and demonstrate that positive lower bound for the two entanglement measures can be obtained if the violation is sufficient. It is very likely that for a given quantum system, choosing different nondegenerate Bell inequalities will cause completely different performance. Therefore, in future work we will try to strengthen the results further. For example, the foundation of our approach is dimension test, and it usually behaves well when the optimal violation corresponds to a maximally entangled state, which may suggest that nondegenerate Bell inequalities like the one introduced in \cite{SAT+17} might be a good choice to quantify entanglement. We hope eventually our new approach can be helpful for experimentalists to quantify entanglement reliably in quantum labs.

\begin{acknowledgments}
Z.W. thanks Valerio Scarani for helpful discussions. Z.W. is supported by the National Key R\&D Program of China, Grant No. 2018YFA0306703 and the start-up funds of Tsinghua University, Grant No. 53330100118. This work has been supported in part by the Zhongguancun Haihua Institute for Frontier Information Technology.
\end{acknowledgments}

\end{document}